%% file: slater.tex
\documentclass[a4paper,USenglish,cleveref,autoref]{lipics-v2021}

\usepackage{tikz}

\title{Determining a Slater Winner is Complete for Parallel Access to NP} 

\titlerunning{Determining a Slater Winner is Complete for Parallel Access to NP}

\author{Michael Lampis}{Université Paris-Dauphine, PSL University, CNRS, LAMSADE, 75016, Paris, France}{michail.lampis@lamsade.dauphine.fr}{https://orcid.org/0000-0002-5791-0887}{}

\authorrunning{Michael Lampis}

\Copyright{Michael Lampis}



\category{}

\relatedversiondetails{Full version}{https://arxiv.org/abs/2103.16416}

\supplement{}

\funding{}

\acknowledgements{I am grateful to Jérôme Lang for letting me know about this problem and for correctly conjecturing that it is complete for $\thtp$.}

\nolinenumbers 

\EventEditors{Petra Berenbrink and Benjamin Monmege}
\EventNoEds{2}
\EventLongTitle{39th International Symposium on Theoretical Aspects of Computer Science (STACS 2022)}
\EventShortTitle{STACS 2022}
\EventAcronym{STACS}
\EventYear{2022}
\EventDate{March 15--18, 2022}
\EventLocation{Marseille, France}
\EventLogo{}
\SeriesVolume{219}
\ArticleNo{5}

\newcommand{\thtp}{\ensuremath \Theta_2^p}

\newcommand{\dout}{\ensuremath d_{\textrm{out}}} \newcommand{\din}{\ensuremath
d_{\textrm{in}}}

\begin{document}

\maketitle

\begin{abstract}

We consider the complexity of deciding the winner of an election under the
Slater rule. In this setting we are given a tournament $T=(V,A)$, where the
vertices of $V$ represent candidates and the direction of each arc indicates
which of the two endpoints is preferable for the majority of voters. The
\emph{Slater score} of a vertex $v\in V$ is defined as the minimum number of
arcs that need to be reversed so that $T$ becomes acyclic and $v$ becomes the
winner. We say that $v$ is a Slater winner in $T$ if $v$ has minimum Slater
score in $T$.

Deciding if a vertex is a Slater winner in a tournament has long been known to
be NP-hard. However, the best known complexity upper bound for this problem is
the class $\thtp$, which corresponds to polynomial-time Turing machines with
parallel access to an NP oracle. In this paper we close this gap by showing
that the problem is $\thtp$-complete, and that this hardness applies to
instances constructible by aggregating the preferences of $7$ voters.  

\end{abstract}

\section{Introduction}

Voting rules, which are a topic of central interest in computational social
choice, are schemes which allow us to aggregate the preferences of a set of
voters among a set of candidates, in order to select a single winner who is
most compatible with the voters' wishes. The main challenge of this area is
that, even if the preferences of each voter are internally consistent (that is,
each voter has a complete ranking of all candidates), it is easy to run into
situations such as the famous Condorcet paradox where collective preferences
are cyclic and hence no clear winner exists. Many rules have therefore been
proposed to deal with this situation and select a winner who is as acceptable
as possible to as many voters as possible. 

In this paper we investigate the computational complexity of a classical and
very natural such voting scheme that is often referred to as the Slater rule.
Intuitively, the idea of the Slater rule is the following: we consider every
possible pair of candidates in our pool $a,b$ and check whether the majority of
voters prefers $a$ or $b$. This allows us to construct a tournament $T$ that
depicts the results of each pariwise matchup between candidates. If $T$ is
transitive (that is, acyclic), then picking a winner is easy. If not, the
Slater rule is that we should select as the winner a candidate who is the
winner of a transitive tournament $T'$ that is at minimum edit distance from
$T$. In other words, a candidate $c$ is a Slater winner if the number of
pairwise matchups that we need to ignore to make $c$ a clear winner is
minimized. 

More formally, the problem we consider is defined as follows. We are given a
set $V$ of $n$ candidates and the preferences of $m$ voters, where each voter's
preferences are given as a total ordering of $V$. We determine a pair-wise
relation on $V$ as follows: for $a,b\in V$ we say that $a$ wins against $b$ if
the majority of voters prefers candidate $a$ over candidate $b$.  In this way,
assuming that there are no ties (which is guaranteed if the number of voters is
odd), we can construct a tournament $T=(V,A)$, where we have the arc $b\to a$
(that is $(b,a)\in A$) if $a$ wins against $b$. In this setting, the Slater
score of a candidate $c$ is the minimum number of arcs of $T$ that need to be
reversed so that $T$ becomes transitive (acyclic) with $c$ being placed last
(that is, with $c$ being a sink). The Slater winner of a tournament is a
candidate with minimum Slater score. Intuitively, a candidate $c$ is a Slater
winner if there exists a linear ordering $\prec$ of the candidates that ranks
$c$ as the winner and is as compatible as possible with the voters' aggregated
preferences, in the sense that the edit distance between $\prec$ and $T$ is
minimum.

The notion of Slater winner is very well-studied and can be seen as a special
case of Kemeny voting. Indeed, in Kemeny voting we construct a weighted
tournament where the weight of the arc $b\to a$ denotes the margin of victory
of $a$ over $b$. In this sense, the Slater system corresponds to a version of
Kemeny voting where we only retain as information which of the two candidates
would win a head-to-head match-up, but ignore the margin of victory.  In other
words, Slater voting is the special case of Kemeny voting where the arcs are
unweighted. For more information about these and other related voting systems,
we refer the reader to \cite{0001CELP16}.

The main question we are interested in in this paper is the computational
complexity of determining if a vertex $v$ of a tournament is a Slater winner.
It has long been known that this question is at least NP-hard \cite{Hudry10}.
Indeed, it is not hard to see that if we had an oracle for the Slater problem
we would be able to produce in polynomial time an ordering of any tournament in
a way that minimizes the number of inversed arcs. This would solve the
\textsc{Feedback Arc Set} problem, which is known to be NP-complete on
tournaments \cite{AilonCN08,Alon06,CharbitTY07,Conitzer06}. On the other hand,
membership of this problem in NP is not obvious. The best currently known upper
bound on its complexity is the class $\thtp$, shown by Hudry
\cite{CharonH10,Hudry10}.

The class $\thtp$ seems like a natural home for the Slater problem. As a
reminder, this class captures as a model of computation  Turing
machines that run in polynomial time and which are allowed to use an NP oracle either a polynomial number of
times non-adaptively (that is, with questions not being allowed to depend on
previous answers), or a logarithmic number of times adaptively.  Hence, this
class is often called ``Parallel Access to NP'' and written as
$\textrm{P}^{\textrm{NP}}_{||}$.  Intuitively, solving the Slater problem
requires us to calculate exactly a value that is NP-hard to compute (the
minimum feedback arc set of a tournament). This can be done either by asking
polynomially many non-adaptive NP queries to an oracle (for each $k=1,2,\ldots$
we ask if the feedback arc set has size at most $k$), or a logarithmic number
of adaptive queries (where we essentially perform binary search). It has
therefore been conjectured that determining if a candidate is a Slater winner
is not just NP-hard, but $\thtp$-complete \cite{0001BH16,0001CELP16,Hudry10}.
We recall that $\thtp$ is strongly suspected to be a much larger class than NP
-- indeed, because $\thtp$ contains all of the so-called Boolean hierarchy of
classes, it is known that if it were the case that $\thtp=NP$, then the
polynomial hierarchy would collapse \cite{ChangK96}.  Hence, the difference
between the known upper and lower bounds on the complexity of determining a
Slater winner is not trivial.

The result we present in this paper settles this problem. We confirm the
conjecture that determining the Slater winner of a tournament is indeed
$\thtp$-complete.  This places Slater voting in the same class as related
voting schemes, such as Kemeny \cite{HemaspaandraSV05}, Dodgson
\cite{HemaspaandraHR97}, and Young \cite{RotheSV03}. It also places it in the
same class as the Slater rule used in \cite{EndrissH15} for a more general
judgment aggregation problem. We prove this result by modifying the reduction
of Conitzer \cite{Conitzer06}, which showed that \textsc{Feedback Arc Set} on
tournaments is NP-complete. The main difference is that, rather than reducing
from \textsc{SAT}, we need to reduce from a $\thtp$-complete variant, where we
are looking for a maximum weight satisfying assignment that sets a certain
variable to True. This forces us to significantly complicate the reduction
because we need to encode in the objective function not only the number of
satisfied clauses but also the weight of the corresponding assignment.

Having settled the worst-case complexity of the problem in general, we go on to
consider a related question: what is the minimum number of voters for which
determining a Slater winner is $\thtp$-complete? The motivation behind this
question is that, even though any tournament can be constructed by aggregating
the preferences of a large enough number of voters\footnote{This is a classical
result known in the literature as McGarvey's theorem.}, if the number of voters
is limited, some tournaments can never arise. Hence the problem may conceivably
be easier if the number of voters is bounded.  In the case of the Slater rule,
Bachmeier et al.\cite{BachmeierBGHKPS19} have shown that determining the Slater
winner remains NP-hard for $7$ voters. By reusing and slightly adjusting their
arguments we improve their complexity lower bound to $\thtp$-completeness for
$7$ voters.

\section{Definitions and Preliminaries}

A tournament is a directed graph $G=(V,A)$ such that for all $x,y\in V$,
exactly one of the arcs $(x,y), (y,x)$ appears in $A$. A feedback arc set (fas)
of a digraph $G=(V,A)$ is a set of arcs $A'\subseteq A$ such that deleting $A'$
from $G$ results in an acyclic digraph. If $G$ is a tournament and $A'$ is a
fas of $G$, then the tournament obtained from $G$ by reversing the direction of
all arcs of $A'$ is acyclic (or transitive). We will say that a total ordering
$\prec$ of the vertices of a digraph $G=(V,A)$ \emph{implies} the fas $S=\{
(x,y)\ |\ (x,y)\in A,\ y\prec x\}$ (in the sense that $S$ is the set of arcs
that disagree with the ordering).  We will say that an ordering of $V$ is
optimal if the fas it implies has minimum size.

Given a digraph $G=(V,A)$ and $v\in V$, we say that $v$ is a Slater winner if
for some $k\ge 0$ the following hold: (i) there exists a fas $S\subseteq A$ of
$G$, such that $v$ is a sink of $G-S$ and $|S|=k$ (ii) every fas of $G$ has
size at least $k$. If $v$ is a Slater winner in $G=(V,A)$, then a winning
ordering for $v$ is a linear ordering of $V$ that places $v$ last and implies a
fas of $G$ of minimum size. 

In a digraph $G=(V,E)$, a set $M\subseteq V$ is a module if the following
holds: for all $x,y\in M$ and $z\not\in M$ we have $(x,z)\in E \leftrightarrow
(y,z)\in E$ and  $(z,x)\in E \leftrightarrow (z,y)\in E$. In other words, every
vertex outside $M$ that has an arc to (respectively from) a vertex of $M$, has
arcs to (respectively from) all of $M$. The following lemma, given by Conitzer
\cite{Conitzer06} with slightly different terminology, states that the vertices
of a module can, without loss of generality, always be ordered together. We say
that the vertices of a set $S$ are contiguous in an ordering $\prec$ if there
are no $x,y\in S$, $z\not\in S$ such that $x\prec z\prec y$.

\begin{lemma}\label{lem:modules} Let $G=(V,A)$ be a digraph, $v\in V$ a vertex,
and suppose we have a partition of $V$ into $k$ non-empty modules $V=M_1\uplus
M_2\uplus\ldots\uplus M_k$.  If $v$ is a Slater winner of $G$, then there
exists a winning ordering for $v$ such that for all $i\in[k]$, the vertices of
$M_i$ are contiguous.\end{lemma}

\begin{proof}

Suppose $k>1$ (otherwise the claim is trivial) and consider an ordering $\prec$
that is winning for $v$. We will say that a set of vertices $S\subseteq V$ is a
\emph{block} of $\prec$ if (i) $S\subseteq M_i$ for some $i\in\{1,\ldots,k\}$;
(ii) $S$ is contiguous; (iii) $S$ is maximal, that is, adding any vertex to $S$
violates one of the two preceding properties.

If the number of blocks is equal to $k$ we are done, as each block is equal to
a module so we have an ordering where each module is contiguous. If we have at
least $k+1$ blocks, we will explain how to edit the ordering so that it remains
winning for $v$, it implies a fas of the same size, and the number of blocks
decreases. Repeating this process until we have $k$ blocks completes the proof.

Consider two vertices $x,y$, with $x\prec y$, which belong to the same module,
say $x,y\in M_1$, but in distinct blocks. Among all such pairs, select $x,y$ so
that their distance in the ordering, that is, the size of the set $Z=\{ z\ |\
x\prec z\prec y\}$ is minimized. Let $X,Y$ be the blocks that contain $x,y$
respectively. Note that by the selection of $x,y$ we have that $x$ is the last
vertex of $X$, $y$ is the first vertex of $Y$, $X\cup Y\subseteq M_1$, and
$M_1\cap Z=\emptyset$.

Let $\dout^Z(x)$ (respectively $\din^Z(x)$) be the out-degree (respectively
in-degree) of $x$ towards the set $Z$. Because $M_1$ is a module, all vertices
of $M_1$ have the same in-degree and out-degree towards $Z$, and in particular,
$\dout^Z(x)=\dout^Z(y)$ and $\din^Z(x)=\din^Z(y)$. Now, if
$\din^Z(x)>\dout^Z(x)$, we can obtain an ordering that implies a smaller fas by
placing $x$ immediately after the last vertex of $Z$. This would contradict the
optimality of $\prec$, so it must be impossible. Similarly, if
$\din^Z(x)<\dout^Z(x)$, we have $\din^Z(y)<\dout^Z(y)$, and we can obtain a
strictly better ordering by placing $y$ immediately before the first vertex of
$Z$, contradiction. We conclude that $\din^Z(x)=\din^Z(y)$. Therefore, moving
all the vertices of $X$ so that they appear immediately after the last vertex
of $Z$ produces an ordering which is equally good as the current one, is still
winning for $v$, and has a smaller number of blocks.  \end{proof}

\subsection{Complexity}

We recall the class $\thtp$ which is known to have several equivalent
characterizations, including $\textrm{P}^{\textrm{NP}[\log n]}$ (P with the
right to make $O(\log n)$ queries to an NP oracle), $\textrm{L}^{\textrm{NP}}$
(logarithmic-space Turing machines with access to an NP oracle), and
$\textrm{P}^{\textrm{NP}}_{||}$ (P with parallel non-adaptive access to an NP
oracle). We refer the reader to \cite{HemaspaandraHR97b} for more information on
this class. In \cite{Wagner87} it was shown that the following problem is
$\thtp$-complete: given a graph $G$, is the maximum clique size $\omega(G)$
odd?  In \cite{Haan19} it is mentioned that the following problem, called
\textsc{Max Model}, is $\thtp$-complete: given a satisfiable CNF formula $\phi$
containing a special variable $x$, is there a satisfying assignment of $\phi$
that sets $x$ to True and has maximum Hamming weight (among all satisfying
assignments), where the Hamming weight of an assignment is the number of
variables it sets to True. 

We will use as a starting point for our reduction a variant of \textsc{Max
Model} which we show is $\thtp$-complete below. The main difference between
this variant and the standard version is that we assume that the given formula
is satisfied by the assignment that sets all variables to False.

\begin{lemma}\label{lem:start} The following problem is $\thtp$-complete. Given
a 3-CNF formula $\phi$ containing a distinguished variable $x$, such that
$\phi$ is satisfied by the all-False assignment, decide if there exists a
satisfying assignment for $\phi$ that sets $x$ to True and has maximum weight
among all satisfying assignments.  \end{lemma}

\begin{proof}

We start with a graph $G=(V,E)$ for which the question is if the maximum
independent set has odd size (clearly this is equivalent to the question of
deciding if the maximum clique has odd size by taking the complement of $G$, so
our starting problem is $\thtp$-complete \cite{Wagner87}).  Let $|V|=n$ and
suppose $V=\{v_1,\ldots,v_n\}$.  We construct a formula $\phi$ as follows: for
each $i\in\{1,\ldots,n\}$ we build $(n+1)$ variables $x_i^1,\ldots,x_i^{n+1}$
and for each $j,k\in \{1,\ldots,n+1\}$ we add the clause $(x_i^j \to x_i^{k})$;
for each $(v_{i_1},v_{i_2})\in E$, for each $j,k\in \{1,\ldots,n+1\}$ we add
the clause $(\neg x_{i_1}^j\lor \neg x_{i_2}^{k})$; we construct $n$ variables
$y_1,\ldots, y_n$ and clauses that represent the constraints $(y_1=x_1^1)$, and
for each $i\in\{2,\ldots,n\}$, $(y_i=y_{i-1}\oplus x_i^1)$. We set $y_n$ as the
distinguished variable of $\phi$.  The formula construction can clearly be
carried out in polynomial time, and no clause has size more than three.
Furthermore, setting everything to False satisfies all clauses. Intuitively,
for each vertex we have constructed $n+1$ variables that will be set to True if
we take this vertex in the independent set. The first set of clauses ensures
that we make a consistent choice among the copies; the second set that we
indeed select an independent set; and the third calculates the parity of its
size.

We now observe that independent sets $S$ of $G$ naturally correspond to
satisfying assignments of $\phi$. In particular, given an independent set
$S\subseteq V$ we can construct an assignment by setting, for all $i,j$,
$x_i^j$ to True if and only if $v_i\in S$; we then complete the assignment by
giving appropriate values to the $y_i$ variables so that the parity constraints
are satisfied. For the converse direction, we can extract an independent set
$S$ from a satisfying assignment by setting $v_i\in S$ if and only if the
assignment sets $x_i^1$ to True. We observe the $y_n$ is set to True in a
satisfying assignment if and only if the corresponding independent set has odd
size (indeed, for each $i$, $y_i$ is set to True if the intersection of the
independent set with the first $i$ vertices has odd size).

Suppose now that there exists an independent set $S$ of maximum size $k$ and
that $k$ is odd. Then, there exists a satisfying assignment of maximum weight
that sets $y_n$ to True. Indeed, suppose for contradiction that the maximum
satisfying assignment $\sigma$ sets $y_n$ to False. Then, the corresponding
independent set $S'$ must have even size $k'$. Since $k$ is odd and $S$ is a
maximum independent set, $k'<k$. But then, the weight of $\sigma$ is at most
$k'(n+1)+n < k(n+1)$. However, the assignment corresponding to $S$ has weight
at least $k(n+1)$, contradiction.

For the converse direction, suppose there exists a satisfying assignment
$\sigma$ of maximum weight that sets $y_n$ to True. The corresponding
independent set $S$ has odd size, say $|S|=k$. If there exists a maximum
independent set $S'$ that has even size $k'$, then $k'>k$. However, the
corresponding truth assignment $\sigma'$ would have weight at least $k'(n+1) >
k(n+1)+n$. Since $\sigma$ has weight at most $k(n+1)+n$ we get a contradiction
to the optimality of $\sigma$.

We conclude that there is a satisfying assignment to $\phi$ of maximum weight
that sets $y_n$ to True if and only if the maximum independent set of $G$ has
odd size. \end{proof}

\section{Reduction to Slater}

This section presents the main result of the paper, stated in Theorem
\ref{thm:slater}. The theorem is based on a reduction from the problem of Lemma
\ref{lem:start} to the problem of deciding if a vertex of a tournament is a
Slater winner. Before we dive into the proof, let us give some high level
intuition (we also invite the reader to take a look at Figure \ref{fig:reduction}).

We will build a tournament to represent a CNF formula $\phi$ with $n$ variables
and $m$ clauses by constructing $n$ groups of ``large'' modules
($A_i,B_i,C_i,D_i,E_i,F_i$, for $i\in \{1,\ldots,n\}$) and $m$ ``small''
modules $T_j$ for $j\in \{1,\ldots,m\}$. The internal structure of the modules
will be irrelevant and we only care about their ordering, which we may assume
to be contiguous thanks to Lemma \ref{lem:modules}.  We will make sure to
adjust the sizes of the modules and their connections so that we have the
following properties:

\begin{enumerate}

\item In any reasonable ordering, all six large modules representing variable $x_i$
come before the six modules representing $x_{i+1}$. This will naturally order
the large modules into $n$ sections.

\item Inside a section, any reasonable ordering will place $A_i, B_i, C_i$
first. Then, if the remaining modules are ordered $D_i\prec E_i\prec F_i$, this
encodes that $x_i$ is set to True.

\item Connections between $T_j$ and variable modules will be such that if $T_j$
is placed completely before or completely after the section of a variable
$x_i$, then the cost is the same. However, the cost may be lower if $T_j$ is
placed inside the section of $x_i$. In that case, we must check if $x_i$
appears in the clause $c_j$ in the original formula and the ordering of the
section of $x_i$ encodes an assignment to $x_i$ that satisfies $c_j$.

\item Variable modules are so large that the ordering must always encode a
satisfying assignment to the formula (which exists by assumption). The ordering
of $T_j$ modules among themselves is irrelevant.

\item In order to encode the weight of a satisfying assignment, we make $E_i$
modules slightly larger (we add $2$ extra vertices). Then, the ordering
$D_i\prec E_i\prec F_i$, which encodes that $x_i$ is True, is better than other
orderings that encode satisfying assignments. Hence, the optimal ordering will
represent a satisfying assignment to $\phi$ with maximum weight.

\item Finally, in order to encode that there is a special variable $x_n$ which
must be set to True, we add one extra vertex to $E_n$. This makes sure that
setting $x_n$ to True is more advantageous than setting any other variable to
True, but not more advantageous than setting two other variables to True.

\end{enumerate}

Armed with the intuition of the previous list, we are now ready to present all
the details of our reduction.

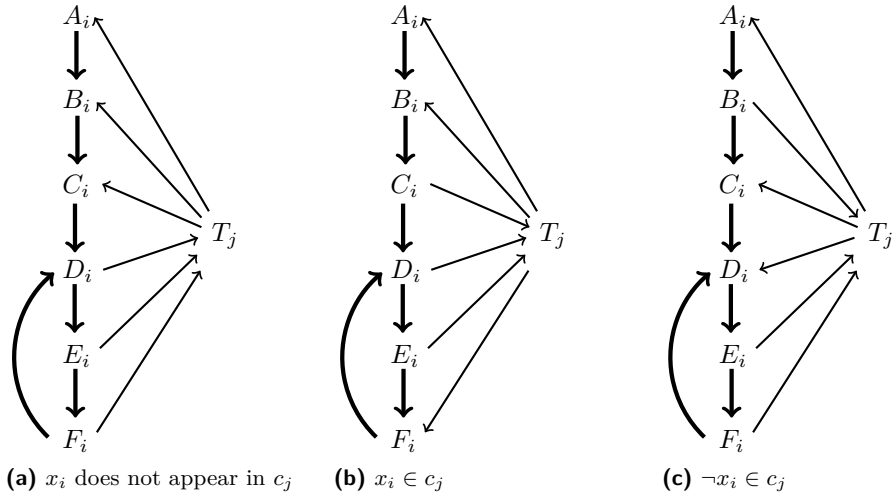
\begin{figure} 
\begin{subfigure}[b]{0.3\textwidth}\input{block2}\caption{$x_i$ does not appear in $c_j$}\end{subfigure}
\begin{subfigure}[b]{0.3\textwidth}\input{block3}\caption{$x_i\in c_j$}\end{subfigure}
\begin{subfigure}[b]{0.3\textwidth}\input{block4}\caption{$\neg x_i\in c_j$}\end{subfigure}
\caption{Gadgets of the reduction of Theorem \ref{thm:slater}. On the left of each figure the six large modules $A_i,B_i,C_i,D_i,E_i,F_i$ represent the variable $x_i$. Missing (thick) arcs go downwards, so the ordering is forced except for the last three modules. The depicted ordering $D_i\prec E_i\prec F_i$ encodes that $x_i$ is True. In the three figures we depict the connections between the six modules and the small module $T_j$ representing clause $c_j$ depending on whether $x_i$ appears in $c_j$. In the first case placing $T_j$ anywhere costs at least three arcs. In the second case, placing $T_j$ after $E_i$ costs two arcs, because setting $x_i$ to True satisfies $c_j$. In the last case, a similarly advantageous placement could be obtained by using the ordering $E_i\prec F_i\prec D_i$ (which encodes that $x_i$ is False) and putting $T_j$ before $D_i$.  }\label{fig:reduction}
\end{figure}

\begin{theorem}\label{thm:slater} The following problem is $\thtp$-complete: given a tournament
$T=(V,A)$ and a vertex $v\in V$, decide if $v$ is a Slater winner.
\end{theorem}

\begin{proof}

We perform a reduction heavily inspired by the reduction of \cite{Conitzer06}
proving that computing the minimum fas of a tournament is NP-complete, though
we include a minor modification proposed by Bachmeier et al.
\cite{BachmeierBGHKPS19} which will later allow us to show that our instances
are realizable using seven voters. The main complication compared to the
reductions of \cite{Conitzer06,BachmeierBGHKPS19} is that we now need to encode
the CNF formula in a way that satisfying assignments of larger weight
correspond to orderings with better objective value.

We start with a formula $\phi$, as given in Lemma \ref{lem:start}.  Let
$x_1,\ldots, x_n$ be the variables of $\phi$ and suppose that the question is
whether there exists a satisfying assignment of maximum weight that sets $x_n$
to True.  Recall that by assumption the all-False assignment satisfies $\phi$.
Let $m$ be the number of clauses of $\phi$.

We define two numbers $s_1,s_2$ which satisfy the following properties:

\begin{eqnarray}
s_1^2 &>& (3n-1)ms_1s_2 + 3ns_1 + m^2s_2^2 + 9m(n-1)s_2 \label{eq:1} \\
s_1s_2 &>& 3ns_1 + m^2s_2^2 + 9m(n-1)s_2 \label{eq:2} \\
s_1 &>& m^2s_2^2 + 9m(n-1)s_2 \label{eq:3} 
\end{eqnarray}

For concreteness, set $s_2 = (n+m)^5$ and $s_1 = s_2^5 = (n+m)^{25}$ and the
above inequalities are easily satisfied when $n+m$ is sufficiently large.
Importantly, $s_1,s_2$ are polynomially bounded in $n+m$. Intuitively, the idea
is that $s_1,s_2$ are two very large numbers, and $s_1$ is significantly
larger. We will construct modules of size (roughly) $s_1$ or $s_2$, and the
values are chosen so that arcs between large modules will be very important,
arcs between large and small modules quite important, and arcs between small
modules almost irrelevant.

We now construct our tournament as follows: for each $i\in \{1,\ldots,n\}$ we
construct $6$ modules, call them $A_i, B_i, C_i, D_i, E_i, F_i$. Modules
$A_i,B_i,C_i,D_i,F_i$ have size $s_1$, while modules $E_i$ have size $s_1+2$ if
$i<n$, and the size of $E_n$ is $s_1+3$. Internally, each of these modules
induces a transitive tournament.  For $i<j$ we add all arcs from $A_i\cup
B_i\cup C_i\cup D_i\cup E_i\cup F_i$ to $A_j\cup B_j\cup C_j\cup D_j\cup
E_j\cup F_j$.  For each $i\in\{1,\ldots,n\}$ we add all possible arcs (i) from
$A_i$ to $B_i\cup C_i\cup D_i\cup E_i\cup F_i$ (ii) from $B_i$ to $C_i\cup
D_i\cup E_i\cup F_i$ (iii) from $C_i$ to $D_i\cup E_i\cup F_i$ (iv) from $D_i$
to $E_i$ (v) from $E_i$ to $F_i$ (vi) from $F_i$ to $D_i$.  The graph we have
constructed so far is a tournament with $n$ sections, each made up of $6$
modules. Each such section represents a variable $x_i$ and the sections are
linearly ordered. The structure inside each section is essentially the
transitive closure of $A_i\to B_i\to C_i\to D_i\to E_i\to F_i$ with the
exception that arcs between $D_i$ and $F_i$ are heading towards $D_i$.

We now complete the construction by adding to the current tournament some
vertices that represent the clauses of $\phi$. In particular, for each
$j\in\{1,\ldots,m\}$ we construct a module $T_j$ of size $s_2$ to represent the
$j$-th clause of $\phi$. Internally, $T_j$ is a transitive tournament. For
$j,j'\in\{1,\ldots,m\}$, the arcs between $T_j$ and $T_{j'}$ are set in an
arbitrary direction. What remains is to explain how the arcs between $T_j$ and
the modules representing the variables are set so as to encode the incidence of
variables with clauses. For each $i\in\{1,\ldots, n\}$ and $j\in\{1,\ldots,m\}$
we do the following:

\begin{enumerate}

\item If $x_i$ does not appear in the $j$-th clause we add all arcs from $T_j$
to $A_i\cup B_i\cup C_i$ and all arcs from $D_i\cup E_i\cup F_i$ to $T_j$.

\item If $x_i$ appears positive in the $j$-th clause we add all arcs from $T_j$
to $A_i\cup B_i \cup F_i$ and all arcs from $C_i\cup D_i\cup E_i$.

\item If $x_i$ appears negative in the $j$-th clause we add all arcs from $T_j$
to $A_i\cup C_i \cup D_i$ and all arcs from $B_i\cup E_i\cup F_i$.

\end{enumerate}

This completes the construction and the question we want to answer is whether
the last vertex (that is, the sink) of the transitive tournament induced by
$F_n$ is a Slater winner of the whole graph. 

We need to prove that the designated vertex is a Slater winner if and only if
there is a satisfying assignment for $\phi$ with maximum weight that sets $x_n$
to True. We will do this by establishing some properties regarding any optimal
ordering of the constructed tournament, showing that such an ordering must
always have a structure which implies a satisfying assignment of $\phi$ with
maximum weight.  We will rely heavily on Lemma \ref{lem:modules}, since the
tournament we have constructed can be decomposed into $6n+m$ modules, namely,
$A_i,B_i,C_i,D_i,E_i$, and $F_i$, for $i\in\{1,\ldots,n\}$, and $T_j$, for
$j\in\{1,\ldots,m\}$. We  therefore assume without loss of generality that
these sets are placed contiguously in an optimal ordering.

Let us first argue that any optimal ordering must have some desirable structure
which necessarily encodes a satisfying assignment for $\phi$. To do this it
will be helpful to start with a baseline ordering and calculate its implied
fas, as then any ordering which implies a larger fas will be necessarily
suboptimal. Consider the ordering which is defined as $A_i \prec B_i \prec C_i
\prec E_i \prec F_i \prec D_i$ for each $i\in\{1,\ldots,n\}$ and which sets
$D_i\prec A_{i+1}$, where each module is internally ordered in the optimal way.
We insert into this ordering of the modules that represent variables, the
modules $T_j$ as follows: for each $j\in\{1,\ldots,m\}$, we find a variable
$x_i$ that appears negative in the $j$-th clause (such a variable must exist,
since $\phi$ is satisfied by the all-False assignment), and place all of $T_j$
between $F_i$ and $D_i$. If for some pair $T_j, T_{j'}$ their relative ordering
is not yet fully specified, we order them in some arbitrary way. 

The arcs incompatible with the above ordering are (i) the at most $ns_1(s_1+3)$
arcs going from a module $D_i$ to a module $E_i$ (ii) for each $T_j$ that was
placed between $F_i$ and $D_i$ we have $2s_1s_2$ arcs (towards $A_i\cup C_i$),
as well as at most $3(s_1+3)s_2$ arcs to each other group $A_{i'}\cup
B_{i'}\cup C_{i'}\cup D_{i'}\cup E_{i'}\cup F_{i'}$, for $i'\neq i$ (iii) the
total number of arcs between modules $T_j$ is at most $m^2s_2^2$. Therefore, we
have that the fas implied by this ordering has size at most

\begin{eqnarray*}
B &\le& ns_1(s_1+3)+ m(2s_1s_2 + 3(n-1)(s_1+3)s_2) + m^2s_2^2= \\
 & =& ns_1^2 + (3n-1)ms_1s_2 + 3ns_1 + m^2s_2^2 +9m(n-1)s_2 
\end{eqnarray*}

In the remainder we will therefore only consider orderings which imply a fas of
size at most $B$, as other orderings are suboptimal. This allows us to draw
some conclusions regarding the structure of an optimal ordering. First, observe
that for each $i\in\{1,\ldots,n\}$, any ordering of $D_i\cup E_i\cup F_i$ will
contribute at least $s_1^2$ arcs to the fas. Using inequality (\ref{eq:1}), we
have that there are at most $n$ pairs of ``large'' modules (that is, modules of
size at least $s_1$) which are incorrectly ordered, that is, ordered so that
all arcs between the modules are included in the fas. Indeed, if there are
$n+1$ such pairs, the fas will have size at least $(n+1)s_1^2>B$. We conclude
that regarding the $6n$ large modules we must have the following ordering:

\begin{enumerate}

\item For each $i<j$, we have that all vertices of $A_i\cup B_i\cup C_i\cup
D_i\cup E_i\cup F_i$ (the section that represents the variable $x_i$) are
before all vertices of $A_j\cup B_j\cup C_j\cup D_j\cup E_j\cup F_j$ (the
section that represents the variable $x_j$).

\item For each $i\in\{1,\ldots,n\}$, we have $A_i\prec B_i\prec C_i$ and all
vertices of $A_i\cup B_i\cup C_i$ are before $D_i\cup E_i\cup F_i$.

\item For each $i\in\{1,\ldots,n\}$ we have $D_i\prec E_i\prec F_i$, or
$E_i\prec F_i\prec D_i$, or $F_i\prec D_i\prec E_i$.

\end{enumerate}

We would now like to construct a correspondence between assignments to $\phi$
and orderings of the tournament that respect the above conditions. On the one
hand, if we are given an assignment $\sigma$ we construct an ordering of the
variable sections as above and for each $i$, if $\sigma$ set $x_i$ to True we
set $D_i\prec E_i\prec F_i$, otherwise we set $E_i\prec F_i\prec D_i$. In the
converse direction, given an ordering that respects the above conditions (which
any optimal ordering must do), we extract an assignment by setting, for each
$i$, $x_i$ to True if  and only if $D_i\prec E_i\prec F_i$.

We now argue that the assignment corresponding to an optimal ordering must also
be satisfying for $\phi$, as otherwise the fas will have size strictly larger
than $B$, contradicting the optimality of the ordering. For the sake of
contradiction, suppose we have an optimal ordering which corresponds to an
assignment falsifying a clause. As argued above, there are at least $ns_1^2$
arcs in the fas contributed by the ordering of the large modules, so we
concentrate on the modules $T_j$ representing clauses. A module $T_j$
representing any clause must be incident on at least $(3n-1)s_1s_2$ arcs of the
fas connecting it to large modules. To see this, consider the following: we
will say that $T_j$ is in the interior of section $i$, if $A_i\prec T_j$ and
$T_j$ is placed before one of $D_i, E_i$, or $F_i$. $T_j$ can be in the
interior of at most one section $i$, so for each $i'\neq i$ we observe that at
least $3s_1s_2$ arcs incident on $T_j$ and modules of the group $i'$ are in the
fas.  This gives $3(n-1)s_1s_2$ arcs. In addition, no matter where we place
$T_j$ in the interior of section $i$, at least a further $2s_1s_2$ arcs of the
fas are obtained: if $T_j$ is after $C_i$, then we get the arcs to $A_i\cup
B_i$ or the arcs to $A_i\cup C_i$; if $T_j$ is between $A_i$ and $C_i$, we get
the arcs to $A_i$ and at least $2s_1s_2$ arcs from $D_i\cup E_i\cup F_i$.
Hence, we get at least $(3n-1)s_1s_2$ arcs in the fas for each $T_j$. 

Furthermore, suppose that the assignment corresponding to the ordering does not
satisfy the $j$-th clause. Then, we claim that at least $3ns_1s_2$ arcs
connecting $T_j$ to large modules are included in the fas. Indeed, if $T_j$ is
in the interior of section $i$, it can either be before or after $C_i$. If it
is before $C_i$, as we observed in the previous paragraph, we always have at
least $3s_1s_2$ arcs in the fas between $T_j$ and the large modules of section
$i$.  If $T_j$ is placed after $C_i$, we have the following cases: (i) if $x_i$
does not appear in the clause, then at least $3s_1s_2$ arcs between $T_j$ and
the large modules of section $i$ are in the fas (ii) if $x_i$ appears positive
in the $j$-th clause, we know that $F_i$ is not placed last in section $i$
(otherwise the assignment would satisfy the $j$-th clause), so wherever we
place $T_j$, at least $3s_1s_2$ arcs are included in the fas (iii) similarly if
$x_i$ appears negative, since the assignment does not satisfy the clause, $F_i$
is last, so again at least $3s_1s_2$ arcs are included in the fas.

From the above calculations, if the assignment that corresponds to an ordering
falsifies a clause, the fas has size at least $ns_1^2 +
(m-1)(3n-1)s_1s_2+3ns_1s_2 = ns_1^2 + (3n-1)ms_1s_2+s_1s_2>B$, where we used
inequality (\ref{eq:2}). We conclude that an optimal ordering must correspond
to a satisfying assignment.

We now need to argue that the assignment corresponding to an optimal ordering
of the tournament must be a satisfying assignment of maximum weight. Suppose
for contradiction that the assignment corresponding to an optimal ordering,
call it $\sigma_1$, sets $k$ variables to True, but there exists another
satisfying assignment, call it $\sigma_2$, that sets at least $k+1$ variables
to True. We will show that starting from $\sigma_2$ we can obtain a better
ordering of the tournament, contradicting the optimality of the original
ordering.

We claim that the ordering from which we extracted $\sigma_1$ includes at least
$ns_1^2 + (3n-1)ms_1s_2 + 2(n-k)s_1$ arcs in the fas. This is because in the
section corresponding to the $n-k$ variables that $\sigma_1$ sets to False,
either the arcs from $E_i$ to $F_i$, or the arcs from $D_i$ to $E_i$ are in the
fas (since $F_i$ is not placed last in the section), and these are at least
$s_1(s_1+2) = s_1^2+2s_1$ arcs.

We construct an ordering from $\sigma_2$ as follows: we order the variable
section in the normal way and inside each section, if
$\sigma_2(x_i)=\textrm{True}$ we use the ordering $D_i\prec E_i\prec F_i$,
otherwise we use the ordering $E_i\prec F_i\prec D_i$. For each $T_j$, we find
a variable $x_i$ that satisfies the $j$-th clause and place $T_j$ in section
$i$ immediately before the last module of this section. If for $j,j'$ the order
of $T_j, T_{j'}$ is not implied by the above, we set it arbitrarily. The fas
implied by this ordering has size at most 

\begin{eqnarray*} 
B' &\le& ns_1^2 + 2(n-k-1)s_1 + s_1 + m(2s_1s_2 + 3(n-1)(s_1+3)s_2) + m^2s_2^2  = \\
 &=& ns_1^2 + 2(n-k)s_1 + (3n-1)ms_1s_2 - s_1 + m^2s_2^2 + 9m(n-1)s_2
\end{eqnarray*}

Here, the calculations for the terms $m(2s_1s_2 + 3(n-1)(s_1+3)s_2) + m^2s_2^2$
are the same as in the calculation of $B$; the term $2(n-k-1)s_1$ takes into
account that there are $n-k-1$ sections that correspond to variables set to
False; and the $s_1$ term is due to the fact that $x_n$ may be one of the
variables set to False and $E_n$ has size $s_1+3$ and not $s_1+2$.  Using
inequality (\ref{eq:3}) we have that $-s_1 +m^2s_2^2+9m(n-1)s_2<0$, so the
ordering we have constructed from $\sigma_2$ is better than the one from which
we extracted $\sigma_1$, contradiction.

At this point we are almost done because we have argued that an optimal
ordering of the tournament corresponds to a satisfying assignment of maximum
weight and furthermore, since the correspondence sets $x_n$ to True if and only
if $F_n$ is the last module in the ordering, the sink of $F_n$ will be last in
the ordering if and only if the assignment sets $x_n$ to True. However, $\phi$
could have several satisfying assignments of the same weight, and since we have
set arcs between $T_j$ modules arbitrarily, it could be the case that a maximum
weight assignment that sets $x_n$ to False results in a better ordering, making
another vertex the Slater winner. This is the reason why we have set $E_n$ to
be slightly larger than all other modules $E_i$, so that setting $x_n$ to True
is always slightly more advantageous than setting any other variable to True.

Concretely, we argue the following: any optimal ordering of the tournament
corresponds to a satisfying assignment of $\phi$ with maximum weight; and
furthermore if a satisfying assignment of $\phi$ with maximum weight sets $x_n$
to True, then any optimal ordering places $F_n$ last. We need to argue the
second claim, so suppose for contradiction that an optimal ordering does not
place $F_n$ last and that the assignment that corresponds to this ordering is
$\sigma_1$. Furthermore, suppose that there exists a satisfying assignment
$\sigma_2$ of maximum weight that sets $x_n$ to True. Say that both $\sigma_1,
\sigma_2$ set $k$ variables to True.

We first observe that the ordering from which we extracted $\sigma_1$ implies a
fas of size at least $ns_1^2 + 2(n-k)s_1 + s_1 + (3n-1)ms_1s_2$. This is
because there are $(n-k-1)$ sections where the fas contains $s_1(s_1+2)$ arcs
incident on a module $E_i$, $k$ sections where the fas contains $s_1^2$ arcs
incident from $F_i$ to $D_i$, and in the section corresponding to $x_n$ the fas
contains $s_1(s_1+3)$ arcs, incident on $E_n$. 

On the other hand, if we construct an ordering from $\sigma_2$ in the same way
as we did previously, the fas obtained will have size at most

\begin{eqnarray*} 
B'' &\le& ns_1^2 + 2(n-k)s_1 + m(2s_1s_2 + 3(n-1)(s_1+3)s_2) + m^2s_2^2  = \\
 &=& ns_1^2 + 2(n-k)s_1 + (3n-1)ms_1s_2 + m^2s_2^2 + 9m(n-1)s_2
\end{eqnarray*}

Again, using inequality (\ref{eq:3}) which states that
$s_1>m^2s_2^2+9m(n-1)s_2$ we conclude that the new ordering is better,
contradicting the optimality of the original ordering.

We now summarize our arguments: we have shown that any optimal ordering of the
tournament always corresponds to a maximum weight satisfying assignmet of
$\phi$ and furthermore, it corresponds to a maximum weight satisfying
assignment that sets $x_n$ to True if this is possible; furthermore, if an
optimal ordering corresponds to an assignment that sets $x_n$ to True then the
last vertex of $F_n$ is a Slater winner. We therefore have two cases: if the
last vertex of $F_n$ is a Slater winner, then since optimal orderings give rise
to satisfying assignments of maximum weight, there is a maximum weight
satisfying assignment of $\phi$ setting $x_n$ to True; if the last vertex of
$F_n$ is not a Slater winner, then the maximum weight satisfying assignment we
extract from an optimal ordering sets $x_n$ to False, and there is no
satisfying assignment of the same weight setting $x_n$ to True.  We conclude
that determining if a vertex is a Slater winner is equivalent to deciding if
$\phi$ has a maximum weight satisfying assignment setting $x_n$ to True, and is
therefore $\thtp$-complete.  \end{proof}

\section{Hardness for $7$ Voters}

In this section we show that the tournaments constructed in Theorem
\ref{thm:slater} correspond to instances that could result from the aggregation
of the preferences of $7$ voters and as a result the problem of determining a
Slater winner remains $\thtp$-complete even for $7$ voters. Our approach
follows along the lines of the arguments of Bachmeier et al.
\cite{BachmeierBGHKPS19} who proved that determining the Slater winner is
NP-hard for $7$ voters. Indeed, the proof of \cite{BachmeierBGHKPS19} consists
of an analysis (and tweak) of the construction of Conitzer \cite{Conitzer06}
which establishes that the instances of the reduction can be built by
aggregating $7$ voter profiles.  Since our reduction is very similar to
Conitzer's, we essentially only need to adjust the arguments of Bachmeier et
al. to obtain $\thtp$-completeness.

Our first step is to slightly restrict the $\thtp$-complete problem that is the
starting point of our reduction. We present the following strengthening of
Lemma \ref{lem:start}, which is similar to the problem used as a starting point
in the reduction of \cite{BachmeierBGHKPS19}.

\begin{lemma}\label{lem:start2} The problem given in Lemma \ref{lem:start}
remains $\thtp$-complete under the following additional restrictions: (i) we
are given a partition of the clauses of $\phi$ in two sets $L,R$ and each
variable appears in at most one clause of $L$ and in at most two clauses of $R$
(ii) each literal appears at most once in a clause of $R$. \end{lemma}

\begin{proof}

Given a formula $\phi$ as in Lemma \ref{lem:start} we construct a new formula
$\phi'$ as follows. Let $x_1,x_2,\ldots,x_n$ be the variables of $\phi$ and $m$
be the number of its clauses. For each $x_i$, $i\in\{1,\ldots,n\}$ we construct
$m$ variables, call them $y_i^1,y_i^2,\ldots, y_i^m$. For each
$i\in\{1,\ldots,n\}$, for $j\in\{1,\ldots,m-1\}$ we construct the clause $(\neg
y_i^j \lor y_i^{j+1})$, as well as the clause $(\neg y_i^{m}\lor y_i^1)$. Let
$R$ be the set of clauses constructed so far and note that each literal appears
at most once and each variable at most twice in these clauses. Intuitively, the
clauses of $R$ ensure that for each $i$, all variables in the set $\{y_i^1,
y_i^2,\ldots, y_i^m\}$ must receive the same value in a satisfying assignment.

Now, we consider the clauses of $\phi$ one by one. If the $j$-th clause
contains the variable $x_i$, we replace it by the variable $y_i^j$. Doing this
for all clauses of $\phi$ we obtain a set of clauses, call it $L$, where each
variable appears at most once (assuming without loss of generality that clauses
of $\phi$ have no repeated literals).

If $x_n$ was the designated variable of $\phi$ we set $y_n^1$ as the designated
variable of $\phi'$. It is now not hard to make a correspondence between
satisfying assignments of $\phi$ and $\phi'$ ($x_i$ is set to True if all
$y_i^j$ are set to True) in a way that preserves weights (the weight of an
assignment to $\phi'$ is $m$ times the weight of the corresponding assignment
for $\phi$). Hence, determining if a maximum weight satisfying assignment to
$\phi'$ sets $y_n^1$ to True is $\thtp$-complete. Observe also that $\phi'$ is
satisfied by the all-False assignment.  \end{proof}

We now obtain the result of this section by starting the reduction of Theorem
\ref{thm:slater} from the problem of Lemma \ref{lem:start2}. In the statement
of the theorem below, when we say that a tournament $T=(V,A)$ can be obtained
from $7$ voters, we mean that there exist $7$ total orderings of $V$ such that
for all $(a,b)\in A$ we have that $a\prec b$ in at least $4$ of the orderings.

\begin{theorem} Determining if a vertex of a tournament is a Slater winner
remains $\thtp$-complete even for tournaments that can be obtained from $7$
voters. \end{theorem}

\begin{proof}

We perform the same reduction as in Theorem \ref{thm:slater} except we start
from the special case given in Lemma \ref{lem:start2}. What remains is to show
that the instance we construct can result from aggregating $7$ orderings.
Recall that our tournament contains $6n$ modules representing the variables,
called $A_i, B_i, C_i, D_i, E_i, F_i$, for $i\in\{1,\ldots,n\}$ and $m$ modules
representing the clauses, called $T_j$, for $j\in\{1,\ldots,m\}$. Since modules
are internally transitive, we will assume that the $7$ voters have preferences
which agree with the directions of the arcs inside the modules and hence we
focus on the arcs between modules. Recall that in the reduction of Theorem
\ref{thm:slater}, arcs between modules $T_j$ are set arbitrarily. To ease
presentation, assume that when $j<j'$ we have the arcs $T_j\to T_{j'}$.

The first voter has preferences $A_1\prec B_1\prec C_1\prec D_1\prec E_1\prec
F_1 \prec A_2 \ldots \prec F_n \prec T_1\prec T_2 \prec\ldots\prec T_m$. In
other words, the first voter orders all the variable modules before all the
clause modules, orders variable groups according to their index, and inside
each variable group she has the ordering $A_i\prec B_i\prec C_i\prec D_i\prec
E_i \prec F_i$.

We now add two voters with the intent of constucting all the arcs of the set \[
X_0 = \left(\bigcup_{j} T_j \times \bigcup_{i} (A_i\cup B_i\cup C_i)\right)
\cup \bigcup_{i} (F_i\times D_i) \]

The first of these voters has ordering $(E_1\prec E_2\prec\ldots\prec E_n)
\prec (F_1\prec D_1\prec F_2\prec D_2\prec\ldots\prec F_n\prec D_n)\prec
(T_1\prec T_2\prec\ldots \prec T_m)\prec (A_1\prec B_1\prec C_1\prec A_2\prec
B_2\prec C_2\prec \ldots \prec A_n\prec B_n\prec C_n)$. The second of these
voters has ordering $(T_m\prec T_{m-1}\prec\ldots\prec T_1)\prec(C_n\prec
B_n\prec A_n\prec C_{n-1}\prec B_{n-1}\prec A_{n-1}\prec\ldots\prec C_1\prec
B_1\prec A_1)\prec (F_n\prec D_n\prec F_{n-1}\prec D_{n-1}\prec\ldots \prec
F_1\prec D_1)\prec(E_n\prec E_{n-1}\prec\ldots\prec E_1)$. Note that these two
voters agree that all modules of $\bigcup_j T_j$ come before all modules of
$\cup_i(A_i\cup B_i\cup C_i)$ and that for each $i$ we have $F_i\prec D_i$, but
disagree on every other pair of modules, hence the two voters together induce
exactly the set of arcs $X_0$ cited above.

If we now consider the three voters we have so far, we observe that much of our
construction is already induced:

\begin{enumerate}

\item For $i<i'$ we have arcs from $A_i\cup B_i\cup C_i\cup D_i\cup E_i\cup
F_i$ to $A_{i'}\cup B_{i'}\cup C_{i'}\cup D_{i'}\cup E_{i'}\cup F_{i'}$ because
of the preferences of the first voter, as the other two voters disagree on
these arcs.

\item For each $i\in\{1,\ldots,n\}$, inside the group $A_i\cup B_i\cup C_i\cup
D_i\cup E_i\cup F_i$, we have arcs that agree with the ordering $A_i\prec
B_i\prec C_i\prec D_i\prec E_i\prec F_i$, except that we have arcs from $F_i$
to $D_i$. This is because the second and third voter agree that $F_i\prec D_i$,
but disagree on every other pair (hence the preferences of the first voter
prevail for the other pairs).

\item For each $j<j'$ we have arcs from $T_j$ to $T_{j'}$, due to the
preferences of the first voter, as the other two disagree.

\item For each $i\in\{1,\ldots,n\}$ and $j\in\{1,\ldots,m\}$ we have arcs from
$T_j$ to $A_i\cup B_i\cup C_i$, because the second and third voter agree that
$T_j\prec (A_i\cup B_i\cup C_i)$ (though the first voter disagrees).

\item For each $i\in\{1,\ldots,n\}$ and $j\in\{1,\ldots,m\}$ we have arcs from
$D_i\cup E_i\cup F_i$ to $T_j$, because the second and third voter disagree on
these pairs, so the preferences of the first voter break the tie.

\end{enumerate}

We therefore have that the tournament that follows from aggregating the
preferences of the first three voters almost corresponds to the one we want to
construct, except that for each $i\in\{1,\ldots,n\}$ and $j\in\{1,\ldots,m\}$
the arcs between $T_j$ and $A_i\cup B_i\cup C_i\cup D_i\cup E_i \cup F_i$
correspond to the arcs we would want if $x_i$ did not appear in the $j$-th
clause (in other words, the three voters we have so far induce the general
structure of the construction, but do not encode which variable appears in
which clause). Furthermore, if we look at the relationship between any two
modules so far, the margin of victory is always exactly one (that is, there do
not exist two modules $X,Y$ such that all three voters agree that $X\prec Y$).

Hence, what remains is to use the four remaining voters to ``fix'' this, so
that if $x_i$ appears (positive or negative) in the $j$-th clause, we have the
arcs prescribed in the reduction of Theorem \ref{thm:slater}. We will achieve
this by giving two pairs of voters. Each pair of voters will disagree on all
pairs of modules except a specific set of arcs that we want to fix. Hence,
adding the pair of voters to the electorate will repair the arcs in question
(since the current margin of victory for all arcs is one), while leaving
everything else unchanged.

Recall that the clause set is given to us partitioned into two sets $R,L$ so
that each variable appears in at most one clause of $L$ and each literal in at
most one clause of $R$. We will use the slightly weaker property that each
literal appears at most once in each of $L,R$. Abusing notation we will write
$j\in R$ if the $j$-th clause is in $R$ (similarly for $j\in L$). We will also
write $x_i\in c_j$ (respectively $\neg x_i\in c_j$) if $x_i$ appears positive
(respectively negative) in the $j$-th clause. 

Consider now the following two sets of arcs:

\[ X_1 = \left(\bigcup_{j\in R} \cup_{i:x_i\in c_j} (T_j\times F_i) \cup
\cup_{i:\neg x_i\in c_j} (T_j\times D_i)\right) \cup \left(\bigcup_{j\in L}
\cup_{i:x_i\in c_j} (C_i\times T_j)\cup \cup_{i:\neg x_i\in c_j} (B_i\times
T_j)\right)\]

\[ X_2 = \left(\bigcup_{j\in L} \cup_{i:x_i\in c_j} (T_j\times F_i) \cup
\cup_{i:\neg x_i\in c_j} (T_j\times D_i)\right) \cup \left(\bigcup_{j\in R}
\cup_{i:x_i\in c_j} (C_i\times T_j)\cup \cup_{i:\neg x_i\in c_j} (B_i\times
T_j)\right)\]

Our plan is to give a pair of voters whose preferences induce the arcs of $X_1$
and another pair whose preferences induce the arcs of $X_2$. Here when we say
that two voters induce a set of arcs $X$ we mean that for each $(a,b)\in X$
both voters have $a\prec b$ and for each $(a,b)\not\in X$ one voter has $a\prec
b$ and the other has $b\prec a$. Before we proceed we observe that if $X_1,X_2$
are inducible by a pair of voters each, then adding these four voters to the
three voters we have described so far produces the tournament of Theorem
\ref{thm:slater}. Indeed, suppose that $x_i$ appears positive in clause $c_j$
and $j\in R$. Then, if we consider the arcs in the tournament induced by the
first three voters, we need to inverse the arcs between $T_j$ and $F_i$ (which
currently point $F_i\to T_j$), and the arcs between $T_j$ and $C_i$ (which
currently point $T_j\to C_i$). But the arcs between $T_j$ and $F_i$ are
inversed thanks to $X_1$, while the arcs between $T_j$ and $C_i$ are inversed
thanks to $X_2$, where we use the fact that $X_1, X_2$ represent the consensus
of two voters, while the margin of victory for any arc induced by the first
three voters is one. Similar arguments apply if $x_i$ appears negative in
$c_j$, or $j\in L$. Hence, if a pair of voters induces $X_1$ and another
induces $X_2$, taking the union of these four voters with the three voters we
have described produces the tournament of Theorem \ref{thm:slater} and
completes the proof.

We now recall that it was shown in \cite{BachmeierBGHKPS19} that $X_1, X_2$ are
inducible by two voters each, since these sets of arcs are unions of stars (if
we contract each module to a vertex).  Let us explain in more detail how to
represent $X_1$ as the union of the preferences of two voters (the arguments
for $X_2$ are essentially identical).  We will make use of the fact that each
literal appears at most once in $L$ and at most once in $R$. 

We will say that a module from a variable group is ``active'' if it is incident
on an arc of $X_1$.  In particular, modules $A_i, E_i$, for
$i\in\{1,\ldots,n\}$ are not active, and neither are modules $F_i$ such that
$x_i$ does not appear positive in $R$ (and similarly for $B_i, C_i, D_i$). We
will concentrate on the ordering of active modules because if we find two voter
profiles that order these modules in a way that induces $X_1$, we can add an
arbitrary ordering of the inactive modules in the beginning of the preferences
of the first voter, and the opposite of that ordering at the end of the
preferences of the second voter.  This will have as effect that the two voters
disagree on any pair that involves an element of an inactive module, as
desired.

We now observe that for each active module $M$ from $\bigcup_i B_i\cup C_i\cup
D_i\cup F_i$, there exists exactly one $T_j$ such that $M$ has arcs to $T_j$ in
$X_1$.  This is because every literal appears in at most one clause of $R$ and
at most one clause of $L$. Now, we construct two voter profiles as follows: one
voter orders the $T_j$ modules in increasing order of index and the other in
decreasing order. For each active module $D_i$ or $F_i$, we insert the module
immediately after the $T_j$ from which the module receives arcs in $X_1$ in
both orderings; for active modules $B_i$ or $C_i$ we insert them immediately
before the $T_j$ towards which the module has arcs in both orderings. Note that
this does not fully specify the ordering, as if two variables $x_i, x_{i'}$
appear in $c_j$ and $j\in R$, then we need to place $F_i$ and $F_{i'}$
immediately after $T_j$.  We resolve such conflicts by using an arbitrary
ordering of the active modules for the first voter and the opposite of that
ordering for the second voter, that is, all modules which are supposed to
appear immediately after $T_j$ are sorted in one way for the first voter and in
the opposite way for the second voter.  We now observe that with this ordering
for every active module the two voters agree about the arcs connecting the
module to its neighboring $T_j$, while we obtain no other arcs between the
module and any other $T_{j'}$ or any other active module. We therefore have two
voters whose preferences induce $X_1$.  \end{proof}

\bibliography{slater}

\end{document}

%% file: block2.tex
\ifx\du\undefined
  \newlength{\du}
\fi
\setlength{\du}{8\unitlength}
\begin{tikzpicture}[even odd rule]
\pgftransformxscale{1.000000}
\pgftransformyscale{-1.000000}
\definecolor{dialinecolor}{rgb}{0.000000, 0.000000, 0.000000}
\pgfsetstrokecolor{dialinecolor}
\pgfsetstrokeopacity{1.000000}
\definecolor{diafillcolor}{rgb}{1.000000, 1.000000, 1.000000}
\pgfsetfillcolor{diafillcolor}
\pgfsetfillopacity{1.000000}
\definecolor{dialinecolor}{rgb}{0.000000, 0.000000, 0.000000}
\pgfsetstrokecolor{dialinecolor}
\pgfsetstrokeopacity{1.000000}
\definecolor{diafillcolor}{rgb}{0.000000, 0.000000, 0.000000}
\pgfsetfillcolor{diafillcolor}
\pgfsetfillopacity{1.000000}
\node[anchor=base west,inner sep=0pt,outer sep=0pt,color=dialinecolor] at (14.000000\du,11.000000\du){$A_i$};
\definecolor{dialinecolor}{rgb}{0.000000, 0.000000, 0.000000}
\pgfsetstrokecolor{dialinecolor}
\pgfsetstrokeopacity{1.000000}
\definecolor{diafillcolor}{rgb}{0.000000, 0.000000, 0.000000}
\pgfsetfillcolor{diafillcolor}
\pgfsetfillopacity{1.000000}
\node[anchor=base west,inner sep=0pt,outer sep=0pt,color=dialinecolor] at (14.000000\du,15.000000\du){$B_i$};
\definecolor{dialinecolor}{rgb}{0.000000, 0.000000, 0.000000}
\pgfsetstrokecolor{dialinecolor}
\pgfsetstrokeopacity{1.000000}
\definecolor{diafillcolor}{rgb}{0.000000, 0.000000, 0.000000}
\pgfsetfillcolor{diafillcolor}
\pgfsetfillopacity{1.000000}
\node[anchor=base west,inner sep=0pt,outer sep=0pt,color=dialinecolor] at (14.000000\du,19.000000\du){$C_i$};
\definecolor{dialinecolor}{rgb}{0.000000, 0.000000, 0.000000}
\pgfsetstrokecolor{dialinecolor}
\pgfsetstrokeopacity{1.000000}
\definecolor{diafillcolor}{rgb}{0.000000, 0.000000, 0.000000}
\pgfsetfillcolor{diafillcolor}
\pgfsetfillopacity{1.000000}
\node[anchor=base west,inner sep=0pt,outer sep=0pt,color=dialinecolor] at (14.000000\du,23.000000\du){$D_i$};
\definecolor{dialinecolor}{rgb}{0.000000, 0.000000, 0.000000}
\pgfsetstrokecolor{dialinecolor}
\pgfsetstrokeopacity{1.000000}
\definecolor{diafillcolor}{rgb}{0.000000, 0.000000, 0.000000}
\pgfsetfillcolor{diafillcolor}
\pgfsetfillopacity{1.000000}
\node[anchor=base west,inner sep=0pt,outer sep=0pt,color=dialinecolor] at (14.000000\du,27.000000\du){$E_i$};
\definecolor{dialinecolor}{rgb}{0.000000, 0.000000, 0.000000}
\pgfsetstrokecolor{dialinecolor}
\pgfsetstrokeopacity{1.000000}
\definecolor{diafillcolor}{rgb}{0.000000, 0.000000, 0.000000}
\pgfsetfillcolor{diafillcolor}
\pgfsetfillopacity{1.000000}
\node[anchor=base west,inner sep=0pt,outer sep=0pt,color=dialinecolor] at (14.000000\du,31.000000\du){$F_i$};
\pgfsetlinewidth{0.200000\du}
\pgfsetdash{}{0pt}
\pgfsetbuttcap
{
\definecolor{diafillcolor}{rgb}{0.000000, 0.000000, 0.000000}
\pgfsetfillcolor{diafillcolor}
\pgfsetfillopacity{1.000000}
\pgfsetarrowsend{to}
\definecolor{dialinecolor}{rgb}{0.000000, 0.000000, 0.000000}
\pgfsetstrokecolor{dialinecolor}
\pgfsetstrokeopacity{1.000000}
\draw (14.650000\du,11.350000\du)--(14.650000\du,13.650000\du);
}
\pgfsetlinewidth{0.200000\du}
\pgfsetdash{}{0pt}
\pgfsetbuttcap
{
\definecolor{diafillcolor}{rgb}{0.000000, 0.000000, 0.000000}
\pgfsetfillcolor{diafillcolor}
\pgfsetfillopacity{1.000000}
\pgfsetarrowsend{to}
\definecolor{dialinecolor}{rgb}{0.000000, 0.000000, 0.000000}
\pgfsetstrokecolor{dialinecolor}
\pgfsetstrokeopacity{1.000000}
\draw (14.696800\du,15.345000\du)--(14.696800\du,17.645000\du);
}
\pgfsetlinewidth{0.200000\du}
\pgfsetdash{}{0pt}
\pgfsetbuttcap
{
\definecolor{diafillcolor}{rgb}{0.000000, 0.000000, 0.000000}
\pgfsetfillcolor{diafillcolor}
\pgfsetfillopacity{1.000000}
\pgfsetarrowsend{to}
\definecolor{dialinecolor}{rgb}{0.000000, 0.000000, 0.000000}
\pgfsetstrokecolor{dialinecolor}
\pgfsetstrokeopacity{1.000000}
\draw (14.581800\du,19.490000\du)--(14.581800\du,21.790000\du);
}
\pgfsetlinewidth{0.200000\du}
\pgfsetdash{}{0pt}
\pgfsetbuttcap
{
\definecolor{diafillcolor}{rgb}{0.000000, 0.000000, 0.000000}
\pgfsetfillcolor{diafillcolor}
\pgfsetfillopacity{1.000000}
\pgfsetarrowsend{to}
\definecolor{dialinecolor}{rgb}{0.000000, 0.000000, 0.000000}
\pgfsetstrokecolor{dialinecolor}
\pgfsetstrokeopacity{1.000000}
\draw (14.566800\du,23.285000\du)--(14.566800\du,25.585000\du);
}
\pgfsetlinewidth{0.200000\du}
\pgfsetdash{}{0pt}
\pgfsetbuttcap
{
\definecolor{diafillcolor}{rgb}{0.000000, 0.000000, 0.000000}
\pgfsetfillcolor{diafillcolor}
\pgfsetfillopacity{1.000000}
\pgfsetarrowsend{to}
\definecolor{dialinecolor}{rgb}{0.000000, 0.000000, 0.000000}
\pgfsetstrokecolor{dialinecolor}
\pgfsetstrokeopacity{1.000000}
\draw (14.601800\du,27.280000\du)--(14.601800\du,29.580000\du);
}
\pgfsetlinewidth{0.200000\du}
\pgfsetdash{}{0pt}
\pgfsetbuttcap
{
\definecolor{diafillcolor}{rgb}{0.000000, 0.000000, 0.000000}
\pgfsetfillcolor{diafillcolor}
\pgfsetfillopacity{1.000000}
\pgfsetarrowsstart{to}
\definecolor{dialinecolor}{rgb}{0.000000, 0.000000, 0.000000}
\pgfsetstrokecolor{dialinecolor}
\pgfsetstrokeopacity{1.000000}
\pgfpathmoveto{\pgfpoint{13.600151\du}{22.749869\du}}
\pgfpatharc{230}{134}{5.220081\du and 5.220081\du}
\pgfusepath{stroke}
}
\definecolor{dialinecolor}{rgb}{0.000000, 0.000000, 0.000000}
\pgfsetstrokecolor{dialinecolor}
\pgfsetstrokeopacity{1.000000}
\definecolor{diafillcolor}{rgb}{0.000000, 0.000000, 0.000000}
\pgfsetfillcolor{diafillcolor}
\pgfsetfillopacity{1.000000}
\node[anchor=base west,inner sep=0pt,outer sep=0pt,color=dialinecolor] at (19.400000\du,19.500000\du){};
\definecolor{dialinecolor}{rgb}{0.000000, 0.000000, 0.000000}
\pgfsetstrokecolor{dialinecolor}
\pgfsetstrokeopacity{1.000000}
\definecolor{diafillcolor}{rgb}{0.000000, 0.000000, 0.000000}
\pgfsetfillcolor{diafillcolor}
\pgfsetfillopacity{1.000000}
\node[anchor=base west,inner sep=0pt,outer sep=0pt,color=dialinecolor] at (20.985000\du,21.254800\du){$T_j$};
\pgfsetlinewidth{0.100000\du}
\pgfsetdash{}{0pt}
\pgfsetbuttcap
{
\definecolor{diafillcolor}{rgb}{0.000000, 0.000000, 0.000000}
\pgfsetfillcolor{diafillcolor}
\pgfsetfillopacity{1.000000}
\pgfsetarrowsstart{to}
\definecolor{dialinecolor}{rgb}{0.000000, 0.000000, 0.000000}
\pgfsetstrokecolor{dialinecolor}
\pgfsetstrokeopacity{1.000000}
\draw (15.500000\du,10.700000\du)--(20.850000\du,19.850000\du);
}
\pgfsetlinewidth{0.100000\du}
\pgfsetdash{}{0pt}
\pgfsetbuttcap
{
\definecolor{diafillcolor}{rgb}{0.000000, 0.000000, 0.000000}
\pgfsetfillcolor{diafillcolor}
\pgfsetfillopacity{1.000000}
\pgfsetarrowsstart{to}
\definecolor{dialinecolor}{rgb}{0.000000, 0.000000, 0.000000}
\pgfsetstrokecolor{dialinecolor}
\pgfsetstrokeopacity{1.000000}
\draw (15.603400\du,14.713400\du)--(20.500000\du,20.150000\du);
}
\pgfsetlinewidth{0.100000\du}
\pgfsetdash{}{0pt}
\pgfsetbuttcap
{
\definecolor{diafillcolor}{rgb}{0.000000, 0.000000, 0.000000}
\pgfsetfillcolor{diafillcolor}
\pgfsetfillopacity{1.000000}
\pgfsetarrowsstart{to}
\definecolor{dialinecolor}{rgb}{0.000000, 0.000000, 0.000000}
\pgfsetstrokecolor{dialinecolor}
\pgfsetstrokeopacity{1.000000}
\draw (15.855600\du,18.565600\du)--(20.500000\du,20.600000\du);
}
\pgfsetlinewidth{0.100000\du}
\pgfsetdash{}{0pt}
\pgfsetbuttcap
{
\definecolor{diafillcolor}{rgb}{0.000000, 0.000000, 0.000000}
\pgfsetfillcolor{diafillcolor}
\pgfsetfillopacity{1.000000}
\pgfsetarrowsend{to}
\definecolor{dialinecolor}{rgb}{0.000000, 0.000000, 0.000000}
\pgfsetstrokecolor{dialinecolor}
\pgfsetstrokeopacity{1.000000}
\draw (15.900900\du,22.610900\du)--(20.350000\du,21.100000\du);
}
\pgfsetlinewidth{0.100000\du}
\pgfsetdash{}{0pt}
\pgfsetbuttcap
{
\definecolor{diafillcolor}{rgb}{0.000000, 0.000000, 0.000000}
\pgfsetfillcolor{diafillcolor}
\pgfsetfillopacity{1.000000}
\pgfsetarrowsend{to}
\definecolor{dialinecolor}{rgb}{0.000000, 0.000000, 0.000000}
\pgfsetstrokecolor{dialinecolor}
\pgfsetstrokeopacity{1.000000}
\draw (15.748400\du,26.309300\du)--(20.300000\du,21.900000\du);
}
\pgfsetlinewidth{0.100000\du}
\pgfsetdash{}{0pt}
\pgfsetbuttcap
{
\definecolor{diafillcolor}{rgb}{0.000000, 0.000000, 0.000000}
\pgfsetfillcolor{diafillcolor}
\pgfsetfillopacity{1.000000}
\pgfsetarrowsend{to}
\definecolor{dialinecolor}{rgb}{0.000000, 0.000000, 0.000000}
\pgfsetstrokecolor{dialinecolor}
\pgfsetstrokeopacity{1.000000}
\draw (15.605700\du,30.275000\du)--(20.450000\du,22.650000\du);
}
\end{tikzpicture}

%% file: block3.tex
\ifx\du\undefined
  \newlength{\du}
\fi
\setlength{\du}{8\unitlength}
\begin{tikzpicture}[even odd rule]
\pgftransformxscale{1.000000}
\pgftransformyscale{-1.000000}
\definecolor{dialinecolor}{rgb}{0.000000, 0.000000, 0.000000}
\pgfsetstrokecolor{dialinecolor}
\pgfsetstrokeopacity{1.000000}
\definecolor{diafillcolor}{rgb}{1.000000, 1.000000, 1.000000}
\pgfsetfillcolor{diafillcolor}
\pgfsetfillopacity{1.000000}
\definecolor{dialinecolor}{rgb}{0.000000, 0.000000, 0.000000}
\pgfsetstrokecolor{dialinecolor}
\pgfsetstrokeopacity{1.000000}
\definecolor{diafillcolor}{rgb}{0.000000, 0.000000, 0.000000}
\pgfsetfillcolor{diafillcolor}
\pgfsetfillopacity{1.000000}
\node[anchor=base west,inner sep=0pt,outer sep=0pt,color=dialinecolor] at (14.000000\du,11.000000\du){$A_i$};
\definecolor{dialinecolor}{rgb}{0.000000, 0.000000, 0.000000}
\pgfsetstrokecolor{dialinecolor}
\pgfsetstrokeopacity{1.000000}
\definecolor{diafillcolor}{rgb}{0.000000, 0.000000, 0.000000}
\pgfsetfillcolor{diafillcolor}
\pgfsetfillopacity{1.000000}
\node[anchor=base west,inner sep=0pt,outer sep=0pt,color=dialinecolor] at (14.000000\du,15.000000\du){$B_i$};
\definecolor{dialinecolor}{rgb}{0.000000, 0.000000, 0.000000}
\pgfsetstrokecolor{dialinecolor}
\pgfsetstrokeopacity{1.000000}
\definecolor{diafillcolor}{rgb}{0.000000, 0.000000, 0.000000}
\pgfsetfillcolor{diafillcolor}
\pgfsetfillopacity{1.000000}
\node[anchor=base west,inner sep=0pt,outer sep=0pt,color=dialinecolor] at (14.000000\du,19.000000\du){$C_i$};
\definecolor{dialinecolor}{rgb}{0.000000, 0.000000, 0.000000}
\pgfsetstrokecolor{dialinecolor}
\pgfsetstrokeopacity{1.000000}
\definecolor{diafillcolor}{rgb}{0.000000, 0.000000, 0.000000}
\pgfsetfillcolor{diafillcolor}
\pgfsetfillopacity{1.000000}
\node[anchor=base west,inner sep=0pt,outer sep=0pt,color=dialinecolor] at (14.000000\du,23.000000\du){$D_i$};
\definecolor{dialinecolor}{rgb}{0.000000, 0.000000, 0.000000}
\pgfsetstrokecolor{dialinecolor}
\pgfsetstrokeopacity{1.000000}
\definecolor{diafillcolor}{rgb}{0.000000, 0.000000, 0.000000}
\pgfsetfillcolor{diafillcolor}
\pgfsetfillopacity{1.000000}
\node[anchor=base west,inner sep=0pt,outer sep=0pt,color=dialinecolor] at (14.000000\du,27.000000\du){$E_i$};
\definecolor{dialinecolor}{rgb}{0.000000, 0.000000, 0.000000}
\pgfsetstrokecolor{dialinecolor}
\pgfsetstrokeopacity{1.000000}
\definecolor{diafillcolor}{rgb}{0.000000, 0.000000, 0.000000}
\pgfsetfillcolor{diafillcolor}
\pgfsetfillopacity{1.000000}
\node[anchor=base west,inner sep=0pt,outer sep=0pt,color=dialinecolor] at (14.000000\du,31.000000\du){$F_i$};
\pgfsetlinewidth{0.200000\du}
\pgfsetdash{}{0pt}
\pgfsetbuttcap
{
\definecolor{diafillcolor}{rgb}{0.000000, 0.000000, 0.000000}
\pgfsetfillcolor{diafillcolor}
\pgfsetfillopacity{1.000000}
\pgfsetarrowsend{to}
\definecolor{dialinecolor}{rgb}{0.000000, 0.000000, 0.000000}
\pgfsetstrokecolor{dialinecolor}
\pgfsetstrokeopacity{1.000000}
\draw (14.650000\du,11.350000\du)--(14.650000\du,13.650000\du);
}
\pgfsetlinewidth{0.200000\du}
\pgfsetdash{}{0pt}
\pgfsetbuttcap
{
\definecolor{diafillcolor}{rgb}{0.000000, 0.000000, 0.000000}
\pgfsetfillcolor{diafillcolor}
\pgfsetfillopacity{1.000000}
\pgfsetarrowsend{to}
\definecolor{dialinecolor}{rgb}{0.000000, 0.000000, 0.000000}
\pgfsetstrokecolor{dialinecolor}
\pgfsetstrokeopacity{1.000000}
\draw (14.696800\du,15.345000\du)--(14.696800\du,17.645000\du);
}
\pgfsetlinewidth{0.200000\du}
\pgfsetdash{}{0pt}
\pgfsetbuttcap
{
\definecolor{diafillcolor}{rgb}{0.000000, 0.000000, 0.000000}
\pgfsetfillcolor{diafillcolor}
\pgfsetfillopacity{1.000000}
\pgfsetarrowsend{to}
\definecolor{dialinecolor}{rgb}{0.000000, 0.000000, 0.000000}
\pgfsetstrokecolor{dialinecolor}
\pgfsetstrokeopacity{1.000000}
\draw (14.581800\du,19.490000\du)--(14.581800\du,21.790000\du);
}
\pgfsetlinewidth{0.200000\du}
\pgfsetdash{}{0pt}
\pgfsetbuttcap
{
\definecolor{diafillcolor}{rgb}{0.000000, 0.000000, 0.000000}
\pgfsetfillcolor{diafillcolor}
\pgfsetfillopacity{1.000000}
\pgfsetarrowsend{to}
\definecolor{dialinecolor}{rgb}{0.000000, 0.000000, 0.000000}
\pgfsetstrokecolor{dialinecolor}
\pgfsetstrokeopacity{1.000000}
\draw (14.566800\du,23.285000\du)--(14.566800\du,25.585000\du);
}
\pgfsetlinewidth{0.200000\du}
\pgfsetdash{}{0pt}
\pgfsetbuttcap
{
\definecolor{diafillcolor}{rgb}{0.000000, 0.000000, 0.000000}
\pgfsetfillcolor{diafillcolor}
\pgfsetfillopacity{1.000000}
\pgfsetarrowsend{to}
\definecolor{dialinecolor}{rgb}{0.000000, 0.000000, 0.000000}
\pgfsetstrokecolor{dialinecolor}
\pgfsetstrokeopacity{1.000000}
\draw (14.601800\du,27.280000\du)--(14.601800\du,29.580000\du);
}
\pgfsetlinewidth{0.200000\du}
\pgfsetdash{}{0pt}
\pgfsetbuttcap
{
\definecolor{diafillcolor}{rgb}{0.000000, 0.000000, 0.000000}
\pgfsetfillcolor{diafillcolor}
\pgfsetfillopacity{1.000000}
\pgfsetarrowsstart{to}
\definecolor{dialinecolor}{rgb}{0.000000, 0.000000, 0.000000}
\pgfsetstrokecolor{dialinecolor}
\pgfsetstrokeopacity{1.000000}
\pgfpathmoveto{\pgfpoint{13.600151\du}{22.749869\du}}
\pgfpatharc{230}{134}{5.220081\du and 5.220081\du}
\pgfusepath{stroke}
}
\definecolor{dialinecolor}{rgb}{0.000000, 0.000000, 0.000000}
\pgfsetstrokecolor{dialinecolor}
\pgfsetstrokeopacity{1.000000}
\definecolor{diafillcolor}{rgb}{0.000000, 0.000000, 0.000000}
\pgfsetfillcolor{diafillcolor}
\pgfsetfillopacity{1.000000}
\node[anchor=base west,inner sep=0pt,outer sep=0pt,color=dialinecolor] at (19.400000\du,19.500000\du){};
\definecolor{dialinecolor}{rgb}{0.000000, 0.000000, 0.000000}
\pgfsetstrokecolor{dialinecolor}
\pgfsetstrokeopacity{1.000000}
\definecolor{diafillcolor}{rgb}{0.000000, 0.000000, 0.000000}
\pgfsetfillcolor{diafillcolor}
\pgfsetfillopacity{1.000000}
\node[anchor=base west,inner sep=0pt,outer sep=0pt,color=dialinecolor] at (20.985000\du,21.254800\du){$T_j$};
\pgfsetlinewidth{0.100000\du}
\pgfsetdash{}{0pt}
\pgfsetbuttcap
{
\definecolor{diafillcolor}{rgb}{0.000000, 0.000000, 0.000000}
\pgfsetfillcolor{diafillcolor}
\pgfsetfillopacity{1.000000}
\pgfsetarrowsstart{to}
\definecolor{dialinecolor}{rgb}{0.000000, 0.000000, 0.000000}
\pgfsetstrokecolor{dialinecolor}
\pgfsetstrokeopacity{1.000000}
\draw (15.500000\du,10.700000\du)--(20.850000\du,19.850000\du);
}
\pgfsetlinewidth{0.100000\du}
\pgfsetdash{}{0pt}
\pgfsetbuttcap
{
\definecolor{diafillcolor}{rgb}{0.000000, 0.000000, 0.000000}
\pgfsetfillcolor{diafillcolor}
\pgfsetfillopacity{1.000000}
\pgfsetarrowsstart{to}
\definecolor{dialinecolor}{rgb}{0.000000, 0.000000, 0.000000}
\pgfsetstrokecolor{dialinecolor}
\pgfsetstrokeopacity{1.000000}
\draw (15.603400\du,14.713400\du)--(20.500000\du,20.150000\du);
}
\pgfsetlinewidth{0.100000\du}
\pgfsetdash{}{0pt}
\pgfsetbuttcap
{
\definecolor{diafillcolor}{rgb}{0.000000, 0.000000, 0.000000}
\pgfsetfillcolor{diafillcolor}
\pgfsetfillopacity{1.000000}
\pgfsetarrowsend{to}
\definecolor{dialinecolor}{rgb}{0.000000, 0.000000, 0.000000}
\pgfsetstrokecolor{dialinecolor}
\pgfsetstrokeopacity{1.000000}
\draw (15.855600\du,18.565600\du)--(20.500000\du,20.600000\du);
}
\pgfsetlinewidth{0.100000\du}
\pgfsetdash{}{0pt}
\pgfsetbuttcap
{
\definecolor{diafillcolor}{rgb}{0.000000, 0.000000, 0.000000}
\pgfsetfillcolor{diafillcolor}
\pgfsetfillopacity{1.000000}
\pgfsetarrowsend{to}
\definecolor{dialinecolor}{rgb}{0.000000, 0.000000, 0.000000}
\pgfsetstrokecolor{dialinecolor}
\pgfsetstrokeopacity{1.000000}
\draw (15.900900\du,22.610900\du)--(20.350000\du,21.100000\du);
}
\pgfsetlinewidth{0.100000\du}
\pgfsetdash{}{0pt}
\pgfsetbuttcap
{
\definecolor{diafillcolor}{rgb}{0.000000, 0.000000, 0.000000}
\pgfsetfillcolor{diafillcolor}
\pgfsetfillopacity{1.000000}
\pgfsetarrowsend{to}
\definecolor{dialinecolor}{rgb}{0.000000, 0.000000, 0.000000}
\pgfsetstrokecolor{dialinecolor}
\pgfsetstrokeopacity{1.000000}
\draw (15.748400\du,26.309300\du)--(20.300000\du,21.900000\du);
}
\pgfsetlinewidth{0.100000\du}
\pgfsetdash{}{0pt}
\pgfsetbuttcap
{
\definecolor{diafillcolor}{rgb}{0.000000, 0.000000, 0.000000}
\pgfsetfillcolor{diafillcolor}
\pgfsetfillopacity{1.000000}
\pgfsetarrowsstart{to}
\definecolor{dialinecolor}{rgb}{0.000000, 0.000000, 0.000000}
\pgfsetstrokecolor{dialinecolor}
\pgfsetstrokeopacity{1.000000}
\draw (15.605700\du,30.275000\du)--(20.450000\du,22.650000\du);
}
\end{tikzpicture}

%% file: block4.tex
\ifx\du\undefined
  \newlength{\du}
\fi
\setlength{\du}{8\unitlength}
\begin{tikzpicture}[even odd rule]
\pgftransformxscale{1.000000}
\pgftransformyscale{-1.000000}
\definecolor{dialinecolor}{rgb}{0.000000, 0.000000, 0.000000}
\pgfsetstrokecolor{dialinecolor}
\pgfsetstrokeopacity{1.000000}
\definecolor{diafillcolor}{rgb}{1.000000, 1.000000, 1.000000}
\pgfsetfillcolor{diafillcolor}
\pgfsetfillopacity{1.000000}
\definecolor{dialinecolor}{rgb}{0.000000, 0.000000, 0.000000}
\pgfsetstrokecolor{dialinecolor}
\pgfsetstrokeopacity{1.000000}
\definecolor{diafillcolor}{rgb}{0.000000, 0.000000, 0.000000}
\pgfsetfillcolor{diafillcolor}
\pgfsetfillopacity{1.000000}
\node[anchor=base west,inner sep=0pt,outer sep=0pt,color=dialinecolor] at (14.000000\du,11.000000\du){$A_i$};
\definecolor{dialinecolor}{rgb}{0.000000, 0.000000, 0.000000}
\pgfsetstrokecolor{dialinecolor}
\pgfsetstrokeopacity{1.000000}
\definecolor{diafillcolor}{rgb}{0.000000, 0.000000, 0.000000}
\pgfsetfillcolor{diafillcolor}
\pgfsetfillopacity{1.000000}
\node[anchor=base west,inner sep=0pt,outer sep=0pt,color=dialinecolor] at (14.000000\du,15.000000\du){$B_i$};
\definecolor{dialinecolor}{rgb}{0.000000, 0.000000, 0.000000}
\pgfsetstrokecolor{dialinecolor}
\pgfsetstrokeopacity{1.000000}
\definecolor{diafillcolor}{rgb}{0.000000, 0.000000, 0.000000}
\pgfsetfillcolor{diafillcolor}
\pgfsetfillopacity{1.000000}
\node[anchor=base west,inner sep=0pt,outer sep=0pt,color=dialinecolor] at (14.000000\du,19.000000\du){$C_i$};
\definecolor{dialinecolor}{rgb}{0.000000, 0.000000, 0.000000}
\pgfsetstrokecolor{dialinecolor}
\pgfsetstrokeopacity{1.000000}
\definecolor{diafillcolor}{rgb}{0.000000, 0.000000, 0.000000}
\pgfsetfillcolor{diafillcolor}
\pgfsetfillopacity{1.000000}
\node[anchor=base west,inner sep=0pt,outer sep=0pt,color=dialinecolor] at (14.000000\du,23.000000\du){$D_i$};
\definecolor{dialinecolor}{rgb}{0.000000, 0.000000, 0.000000}
\pgfsetstrokecolor{dialinecolor}
\pgfsetstrokeopacity{1.000000}
\definecolor{diafillcolor}{rgb}{0.000000, 0.000000, 0.000000}
\pgfsetfillcolor{diafillcolor}
\pgfsetfillopacity{1.000000}
\node[anchor=base west,inner sep=0pt,outer sep=0pt,color=dialinecolor] at (14.000000\du,27.000000\du){$E_i$};
\definecolor{dialinecolor}{rgb}{0.000000, 0.000000, 0.000000}
\pgfsetstrokecolor{dialinecolor}
\pgfsetstrokeopacity{1.000000}
\definecolor{diafillcolor}{rgb}{0.000000, 0.000000, 0.000000}
\pgfsetfillcolor{diafillcolor}
\pgfsetfillopacity{1.000000}
\node[anchor=base west,inner sep=0pt,outer sep=0pt,color=dialinecolor] at (14.000000\du,31.000000\du){$F_i$};
\pgfsetlinewidth{0.200000\du}
\pgfsetdash{}{0pt}
\pgfsetbuttcap
{
\definecolor{diafillcolor}{rgb}{0.000000, 0.000000, 0.000000}
\pgfsetfillcolor{diafillcolor}
\pgfsetfillopacity{1.000000}
\pgfsetarrowsend{to}
\definecolor{dialinecolor}{rgb}{0.000000, 0.000000, 0.000000}
\pgfsetstrokecolor{dialinecolor}
\pgfsetstrokeopacity{1.000000}
\draw (14.650000\du,11.350000\du)--(14.650000\du,13.650000\du);
}
\pgfsetlinewidth{0.200000\du}
\pgfsetdash{}{0pt}
\pgfsetbuttcap
{
\definecolor{diafillcolor}{rgb}{0.000000, 0.000000, 0.000000}
\pgfsetfillcolor{diafillcolor}
\pgfsetfillopacity{1.000000}
\pgfsetarrowsend{to}
\definecolor{dialinecolor}{rgb}{0.000000, 0.000000, 0.000000}
\pgfsetstrokecolor{dialinecolor}
\pgfsetstrokeopacity{1.000000}
\draw (14.696800\du,15.345000\du)--(14.696800\du,17.645000\du);
}
\pgfsetlinewidth{0.200000\du}
\pgfsetdash{}{0pt}
\pgfsetbuttcap
{
\definecolor{diafillcolor}{rgb}{0.000000, 0.000000, 0.000000}
\pgfsetfillcolor{diafillcolor}
\pgfsetfillopacity{1.000000}
\pgfsetarrowsend{to}
\definecolor{dialinecolor}{rgb}{0.000000, 0.000000, 0.000000}
\pgfsetstrokecolor{dialinecolor}
\pgfsetstrokeopacity{1.000000}
\draw (14.581800\du,19.490000\du)--(14.581800\du,21.790000\du);
}
\pgfsetlinewidth{0.200000\du}
\pgfsetdash{}{0pt}
\pgfsetbuttcap
{
\definecolor{diafillcolor}{rgb}{0.000000, 0.000000, 0.000000}
\pgfsetfillcolor{diafillcolor}
\pgfsetfillopacity{1.000000}
\pgfsetarrowsend{to}
\definecolor{dialinecolor}{rgb}{0.000000, 0.000000, 0.000000}
\pgfsetstrokecolor{dialinecolor}
\pgfsetstrokeopacity{1.000000}
\draw (14.566800\du,23.285000\du)--(14.566800\du,25.585000\du);
}
\pgfsetlinewidth{0.200000\du}
\pgfsetdash{}{0pt}
\pgfsetbuttcap
{
\definecolor{diafillcolor}{rgb}{0.000000, 0.000000, 0.000000}
\pgfsetfillcolor{diafillcolor}
\pgfsetfillopacity{1.000000}
\pgfsetarrowsend{to}
\definecolor{dialinecolor}{rgb}{0.000000, 0.000000, 0.000000}
\pgfsetstrokecolor{dialinecolor}
\pgfsetstrokeopacity{1.000000}
\draw (14.601800\du,27.280000\du)--(14.601800\du,29.580000\du);
}
\pgfsetlinewidth{0.200000\du}
\pgfsetdash{}{0pt}
\pgfsetbuttcap
{
\definecolor{diafillcolor}{rgb}{0.000000, 0.000000, 0.000000}
\pgfsetfillcolor{diafillcolor}
\pgfsetfillopacity{1.000000}
\pgfsetarrowsstart{to}
\definecolor{dialinecolor}{rgb}{0.000000, 0.000000, 0.000000}
\pgfsetstrokecolor{dialinecolor}
\pgfsetstrokeopacity{1.000000}
\pgfpathmoveto{\pgfpoint{13.600151\du}{22.749869\du}}
\pgfpatharc{230}{134}{5.220081\du and 5.220081\du}
\pgfusepath{stroke}
}
\definecolor{dialinecolor}{rgb}{0.000000, 0.000000, 0.000000}
\pgfsetstrokecolor{dialinecolor}
\pgfsetstrokeopacity{1.000000}
\definecolor{diafillcolor}{rgb}{0.000000, 0.000000, 0.000000}
\pgfsetfillcolor{diafillcolor}
\pgfsetfillopacity{1.000000}
\node[anchor=base west,inner sep=0pt,outer sep=0pt,color=dialinecolor] at (19.400000\du,19.500000\du){};
\definecolor{dialinecolor}{rgb}{0.000000, 0.000000, 0.000000}
\pgfsetstrokecolor{dialinecolor}
\pgfsetstrokeopacity{1.000000}
\definecolor{diafillcolor}{rgb}{0.000000, 0.000000, 0.000000}
\pgfsetfillcolor{diafillcolor}
\pgfsetfillopacity{1.000000}
\node[anchor=base west,inner sep=0pt,outer sep=0pt,color=dialinecolor] at (20.985000\du,21.254800\du){$T_j$};
\pgfsetlinewidth{0.100000\du}
\pgfsetdash{}{0pt}
\pgfsetbuttcap
{
\definecolor{diafillcolor}{rgb}{0.000000, 0.000000, 0.000000}
\pgfsetfillcolor{diafillcolor}
\pgfsetfillopacity{1.000000}
\pgfsetarrowsstart{to}
\definecolor{dialinecolor}{rgb}{0.000000, 0.000000, 0.000000}
\pgfsetstrokecolor{dialinecolor}
\pgfsetstrokeopacity{1.000000}
\draw (15.500000\du,10.700000\du)--(20.850000\du,19.850000\du);
}
\pgfsetlinewidth{0.100000\du}
\pgfsetdash{}{0pt}
\pgfsetbuttcap
{
\definecolor{diafillcolor}{rgb}{0.000000, 0.000000, 0.000000}
\pgfsetfillcolor{diafillcolor}
\pgfsetfillopacity{1.000000}
\pgfsetarrowsend{to}
\definecolor{dialinecolor}{rgb}{0.000000, 0.000000, 0.000000}
\pgfsetstrokecolor{dialinecolor}
\pgfsetstrokeopacity{1.000000}
\draw (15.603400\du,14.713400\du)--(20.500000\du,20.150000\du);
}
\pgfsetlinewidth{0.100000\du}
\pgfsetdash{}{0pt}
\pgfsetbuttcap
{
\definecolor{diafillcolor}{rgb}{0.000000, 0.000000, 0.000000}
\pgfsetfillcolor{diafillcolor}
\pgfsetfillopacity{1.000000}
\pgfsetarrowsstart{to}
\definecolor{dialinecolor}{rgb}{0.000000, 0.000000, 0.000000}
\pgfsetstrokecolor{dialinecolor}
\pgfsetstrokeopacity{1.000000}
\draw (15.855600\du,18.565600\du)--(20.500000\du,20.600000\du);
}
\pgfsetlinewidth{0.100000\du}
\pgfsetdash{}{0pt}
\pgfsetbuttcap
{
\definecolor{diafillcolor}{rgb}{0.000000, 0.000000, 0.000000}
\pgfsetfillcolor{diafillcolor}
\pgfsetfillopacity{1.000000}
\pgfsetarrowsstart{to}
\definecolor{dialinecolor}{rgb}{0.000000, 0.000000, 0.000000}
\pgfsetstrokecolor{dialinecolor}
\pgfsetstrokeopacity{1.000000}
\draw (15.900900\du,22.610900\du)--(20.350000\du,21.100000\du);
}
\pgfsetlinewidth{0.100000\du}
\pgfsetdash{}{0pt}
\pgfsetbuttcap
{
\definecolor{diafillcolor}{rgb}{0.000000, 0.000000, 0.000000}
\pgfsetfillcolor{diafillcolor}
\pgfsetfillopacity{1.000000}
\pgfsetarrowsend{to}
\definecolor{dialinecolor}{rgb}{0.000000, 0.000000, 0.000000}
\pgfsetstrokecolor{dialinecolor}
\pgfsetstrokeopacity{1.000000}
\draw (15.748400\du,26.309300\du)--(20.300000\du,21.900000\du);
}
\pgfsetlinewidth{0.100000\du}
\pgfsetdash{}{0pt}
\pgfsetbuttcap
{
\definecolor{diafillcolor}{rgb}{0.000000, 0.000000, 0.000000}
\pgfsetfillcolor{diafillcolor}
\pgfsetfillopacity{1.000000}
\pgfsetarrowsend{to}
\definecolor{dialinecolor}{rgb}{0.000000, 0.000000, 0.000000}
\pgfsetstrokecolor{dialinecolor}
\pgfsetstrokeopacity{1.000000}
\draw (15.605700\du,30.275000\du)--(20.450000\du,22.650000\du);
}
\end{tikzpicture}